\pdfoutput=1 

\documentclass[conference]{IEEEtran} 
\IEEEoverridecommandlockouts
\usepackage[english]{babel}
\usepackage{graphicx}
\usepackage[latin1]{inputenc}
\usepackage[caption=false]{subfig}
\usepackage{amsmath,amssymb,amsthm,latexsym}
\usepackage{algorithm}
\usepackage{algpseudocode}
\usepackage{xspace}
\usepackage[svgnames]{xcolor}
\usepackage[bookmarks=false,draft=true]{hyperref}
\usepackage{cite}

\def\nd{\textsuperscript{nd}}
\def\rd{\textsuperscript{rd}}
\def\th{\textsuperscript{th}}

\newtheorem{theorem}{Theorem}
\newtheorem{lemma}[theorem]{Lemma}

\newtheorem{corollary}[theorem]{Corollary}

\def\genff{\omega}

\DeclareMathOperator{\rank}{rank}
\DeclareMathOperator{\colspan}{colspan}

\def\MM{\mathcal{M}}

\def\vec#1{\mathbf{#1}}

\def\matdesc#1(#2){
\left[\begin{array}{c}#1\end{array}\right]#2
}

\hypersetup{%
pdftitle={Exact Scalar Minimum Storage Coordinated Regenerating Codes},
pdfauthor={Nicolas Le Scouarnec},
pdfkeywords={network coding, regenerating codes, exact repair, multiple failures, coordinated},
pdfstartview={FitH},
urlcolor=darkblue,
pdfborder=0 0 0,
bookmarksnumbered
}%

\begin{document}
\sloppy

\definecolor{darkgreen}{rgb}{0,0.5,0}
\def\draft#1{\textcolor{darkgreen}{#1}}

\makeatletter
\def\ps@headings{%
\def\@oddhead{\mbox{}\scriptsize\rightmark \hfil \thepage}%
\def\@evenhead{\scriptsize\thepage \hfil \leftmark\mbox{}}%
\def\@oddfoot{}%
\def\@evenfoot{}}
\makeatother
\pagestyle{headings}

\title{Exact Scalar Minimum Storage\\ Coordinated Regenerating Codes}

\author{
	Nicolas Le Scouarnec\\
	Technicolor\\
	Rennes, France
}
\maketitle

\begin{abstract}
We study the exact and optimal repair of multiple failures in codes for distributed storage. More particularly, we examine the use of interference alignment to build exact scalar minimum storage coordinated regenerating codes (MSCR). We  show that it is possible to build codes for the case of $k=2$ and $d \ge k$ by aligning interferences independently but that this technique cannot be applied as soon as  $k \ge 3$ and $d>k$. Our results also apply to adaptive regenerating codes.
\end{abstract}

\section{Introduction}
Codes allow to implement redundancy in distributed storage systems so that device failures  do no hurt the whole system. Yet, to keep preventing failures, once failures have occurred, codes must be repaired: the redundancy level must be kept above some minimum level. The naïve approach to repairing codes consists in decoding the whole code (thus downloading all blocks) so as to encode it again to recreate the few lost blocks. This induces huge repair costs in term of network bandwidth. It has recently been shown that this repair cost can be significantly reduced by repairing without decoding using regenerating codes. Lower bounds on costs (\emph{i.e.}, tradeoffs between storage and bandwidth)  have been established for both the single failure case~\cite{Dimakis2007,Dimakis2010},  and the multiple failures case~\cite{NetCod2011,Hu2010,Shum2011}. Adaptive regenerating codes, departing from the other studies by allowing the number of devices involved to differ between repairs, have been defined in~\cite{NetCod2011}. The two extreme points of the optimal tradeoffs are  Minimum Bandwidth (MBR/MBCR), which minimizes repair cost first, and Minimum Storage (MSR/MSCR), which minimizes storage first.  Codes matching these theoretical tradeoffs can be built using non-deterministic schemes such as random linear network codes.

However, non-deterministic schemes for regenerating codes are not desiderable since they \emph{(i)} require a great field size, \emph{(ii)} require homomorphic hash functions to provide basic security (integrity checking), \emph{(iii)} cannot be turned into systematic codes, which offer access to data without decoding, and \emph{(iv)} provide only probabilistic guarantees. Deterministic schemes overcome these issues by offering exact repair (\emph{i.e.}, during a repair, the regenerated block is equal to the lost block and not only \emph{equivalent}). For the single failure case ($t=1$), code constructions with exact repair have been given for both the MSR point ($n,k,d\ge 2k-2$ ~\cite{Rashmi2011} and $n,k,d$ when the size of the file is infinite~\cite{Cadambe2010,Suh2010a}) and the MBR point ($n,k,d$~\cite{Rashmi2011}) where $n$ is the number of encoded blocks, $k$ is the number of original blocks, and $d$ is the number of devices contacted during repairs. Recent works on this problem are surveyed in~\cite{Dimakis2010b}. However, the existence of codes supporting the exact repair of multiple failures ($t > 1$) (\emph{i.e.}, exact coordinated/adaptive regenerating codes) is an open question.

In this paper, we focus on this problem, thus extending our previous work on coordinated regenerating codes in~\cite{NetCod2011} with exact repair. We consider the case of $n,k,d > k,t > 1$ for scalar constructions (\emph{i.e.}, $\beta=1$)  and make the following contributions:
\begin{itemize}
\item In the line of  exact scalar minimum storage regenerating codes~\cite{Shah2010b,Rashmi2011,Suh2011}, we propose exact scalar minimum storage coordinated regenerating codes (MSCR) for the case $n,k=2,d \ge k,t=n-d$.  This interference alignment based construction is inspired by \cite{Shah2010b,Suh2011}. (Section~\ref{sec:code})
\item Inteference alignment has been applied to scalar MSR codes by aligning the various interferences independently. We show that when  $k\ge{}3$, aligning interferences independently,  as in~\cite{Suh2011,Shah2010b}, is not sufficient to repair exactly scalar MSCR codes. (Section~\ref{sec:nonachiev}).
\end{itemize} 
Note that these results, which correspond to the MSCR point, also apply to exact scalar adapative regenerating codes~\cite{NetCod2011}.

As explained earlier, most previous works have been limited to single failures ($t=1$).  For the multiple failures, there only exist results for the case $n,k,d=k,t=n-k$, a degenerated case where the repair of regenerating codes and the naïve approach to repairing erasure correcting codes are the same. In this case, the exact repair of MSCR boils down to performing, in parallel, the repair of $t$ independent erasure correcting codes~\cite{Shum2011}. A similar construction exists for MBCR codes~\cite{Shum2011b}. The position of our codes among existing codes constructions is detailled in Section~\ref{sec:related}.

\section{Background}

\begin{figure*}[t]
\centering
\subfloat[Functional repairs]{\centering
        \includegraphics[height=0.17\linewidth,trim=0 -35 0 0]{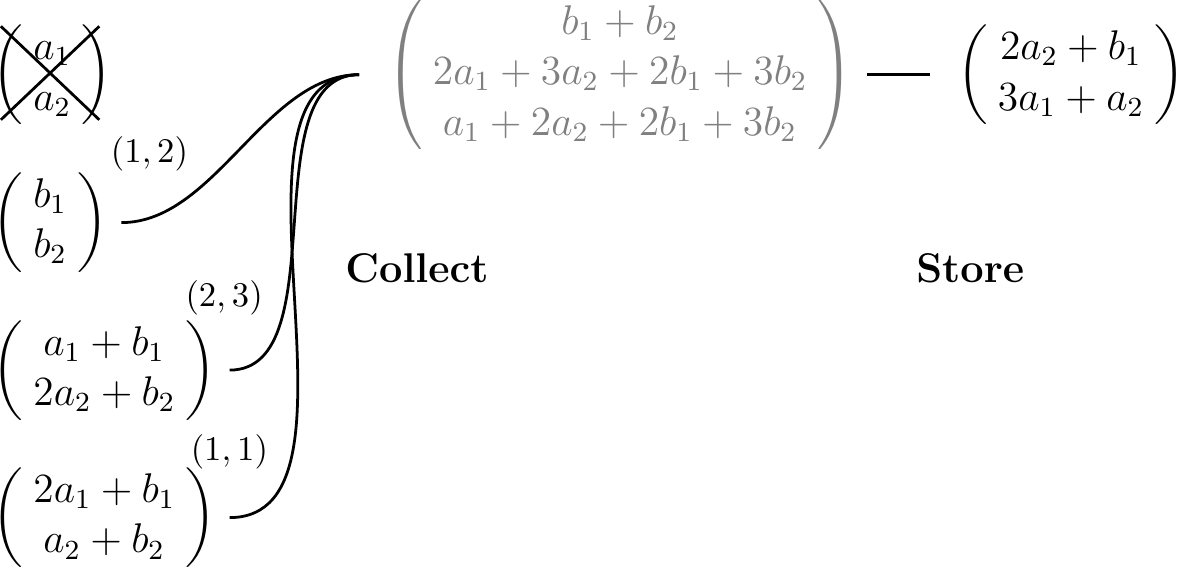}
        \label{fig:functional}
} \hfil
\subfloat[Exact repairs (scalar $\beta=1$)]{\centering
        \includegraphics[height=0.17\linewidth,trim=0 -35 0 0]{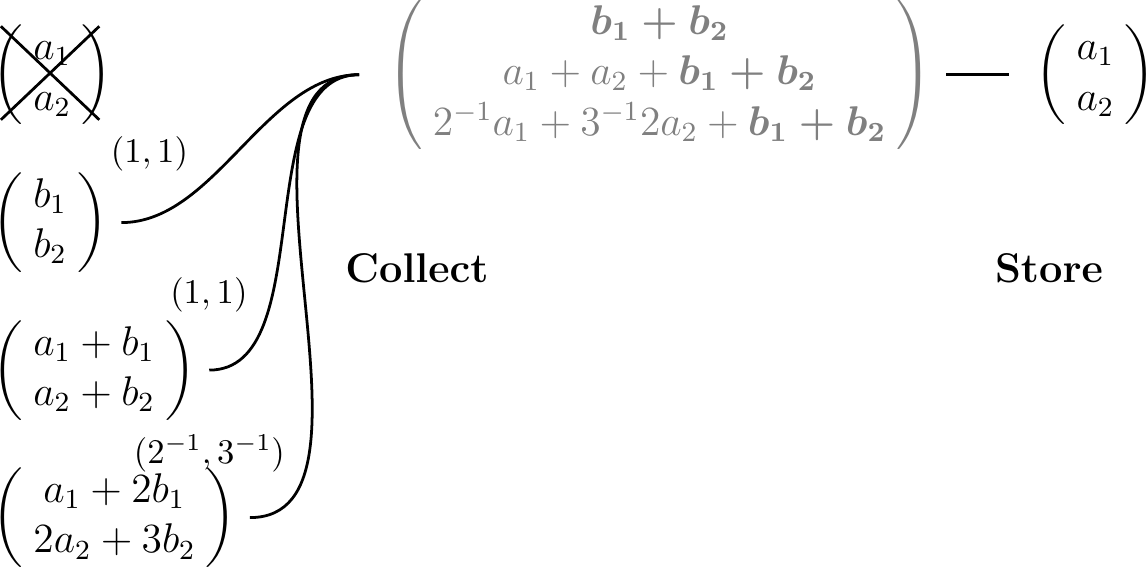}
        \label{fig:exact}
}\hfil
\subfloat[Repair (vector $\beta=2$)]{\centering
        \includegraphics[height=0.18\linewidth]{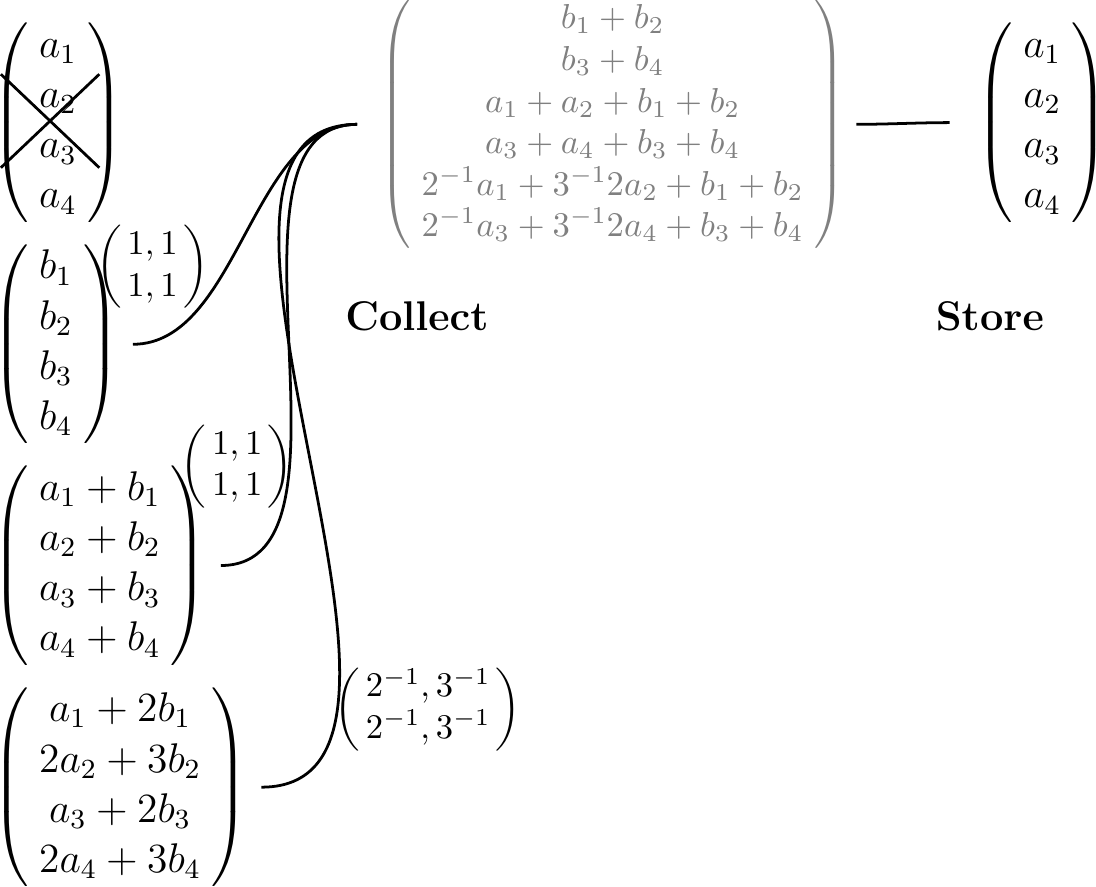}
        \label{fig:sub}
}
\caption{Regenerating codes can be repaired functionally or exactly. In our example, the device storing $(a_1,a_2)$ fails and is regenerated. When relying on functional repairs, the information about $(a_1,a_2)$ is regenerated but not in the same form, while when relying on exact repairs, $(a_1,a_2)$ is regenerated exactly. This figure also illustrates the difference between scalar codes where scalar are transmitted over the network and vector codes where vectors are sent over the network.}
\label{fig:soa}
\end{figure*}

We consider a $n$ devices system storing a file of $\MM$ bits. The file is encoded and dispatched accross all $n$ devices (each storing $\alpha$ bits) so that the file can be recovered by collecting data from any $k$ devices. Whenever  devices fail, they must be repaired so that the level of redundancy does not fall bellow a critical level.  Classical erasure correcting codes require a decoding to be performed to repair any single lost block by encoding the decoded data and dispatching again. This approach has huge repair costs (in term of network communications). It has been shown that this cost can be significantly reduced by relying on regenerating codes~\cite{Dimakis2007,Dimakis2010}. Similar results have been given for repairing multiple failures using coordinated/cooperative regenerating codes~\cite{NetCod2011,Hu2010,Shum2011}.

For repairing coordinated regenerating codes, each failed device\footnote{In the article, we use \emph{failed devices} to designate either the devices that have failed, or the new spare devices that holds the repaired data. The meaning will be clear from the context.} contacts $d \ge k$ live devices and gets $\beta$ bits from each. The $t$ failed devices coordinate by exchanging $\beta'$ bits. The data is then processed and $\alpha$ bits are stored.
The amounts of data exchanged and stored during repairs are summarized on Figure~\ref{fig:repair_rc}.
These studies lead to the definition of the optimal tradeoffs between storage $\alpha$ and repair costs $\gamma=d\beta+(t-1)\beta'$. The two extreme points of the optimal tradeoffs are shown on Figure~\ref{fig:cr} with the corresponding values of $\alpha$, $\beta$ and $\beta'$. The MSCR (resp. MBCR) point minimizes storage (resp. bandwidth) first. In this paper, we will focus on MSCR constructions for they are very close to classical erasure correcting codes and are highly related to adaptive regenerating codes.

\begin{figure}[h]
\centering
        \includegraphics[width=0.68\linewidth]{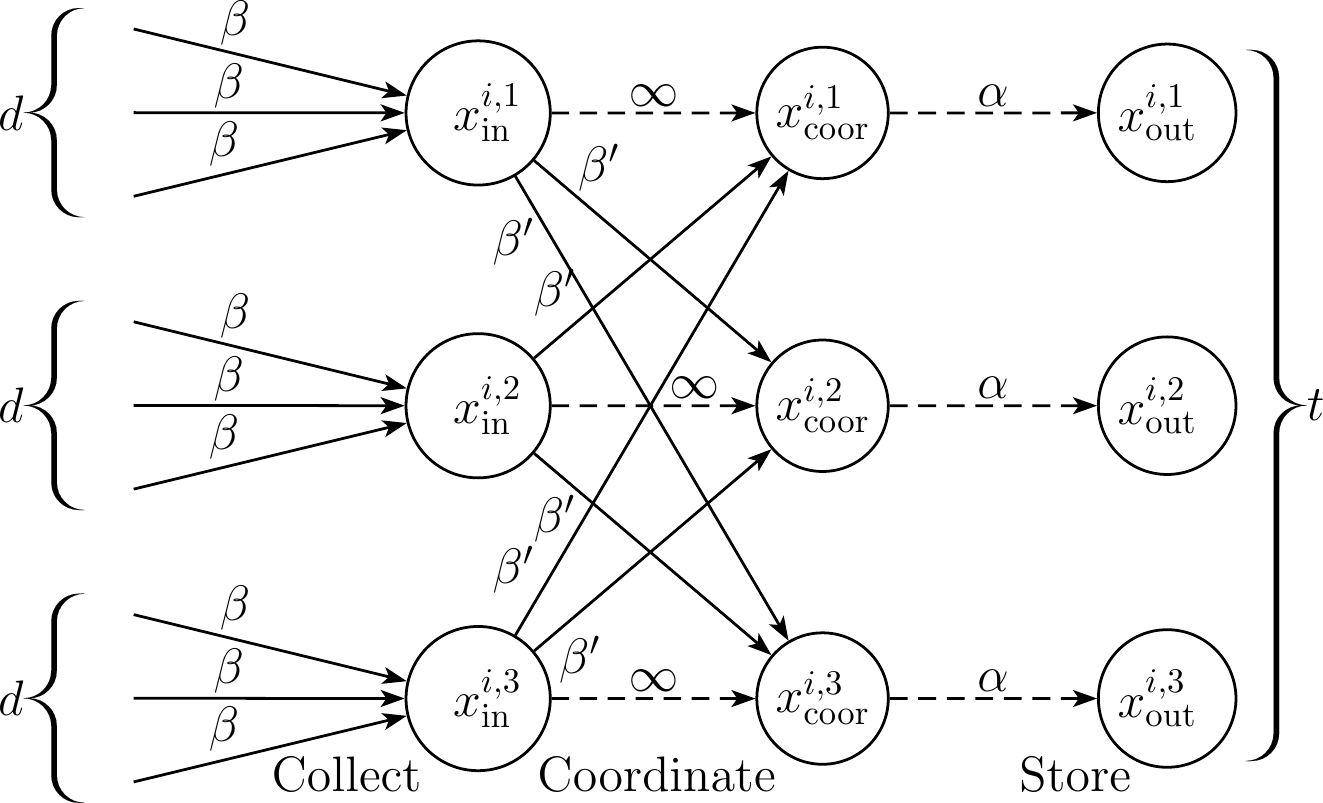}
        \caption{Amounts of information exchanged during the repair of $t$ failed devices from $d$ live devices (an infinite capacity means that all the information received is kept for later processing). During a first step, each failed device collects $\beta$ bits from $d$ live devices. All failed devices coordinate by exchanging $\beta'$ bits. The data is processed and $\alpha$ bits are stored. Solid lines show transfers over the network.}
        \label{fig:repair_rc}
\end{figure}

These tradeoffs are derived from network coding results, through a reduction to a multicast problem. Hence, non-determinstic coding schemes matching these tradeoffs  can be built using random linear network codes. The corresponding non-deterministic repairs are termed as functional repairs. Yet, such codes have several disadvantages: \emph{(i)} they have high decoding costs, \emph{(ii)} they make the implementation of integrity checking complex by requiring the use of homomorphic hashes, \emph{(iii)} they cannot be turned into systematic codes, which provide access to data without decoding, and \emph{(iv)} they can only provide probabilistic guarantees.

\begin{figure}[htbp]
	\centering%
	\includegraphics[width=0.85\linewidth]{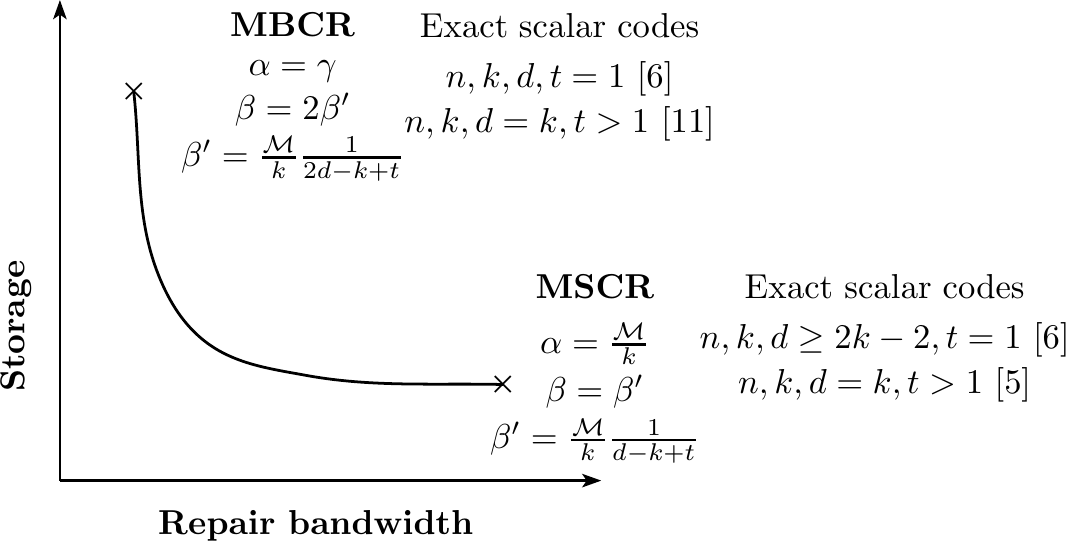}%
	\caption{Regenerating codes achieve the optimal tradeoff between storage and bandwidth (\emph{i.e.}, repair cost). The figure shows the values for the MSCR and MBCR points. The figure also shows the best exact scalar regenerating codes ($\beta=1$) known for the single ($t=1$) and multiple failure cases ($t>1$).}%
	\label{fig:cr}%
\end{figure}

\begin{figure*}[t]
	\centering%
	\includegraphics[width=0.9\linewidth]{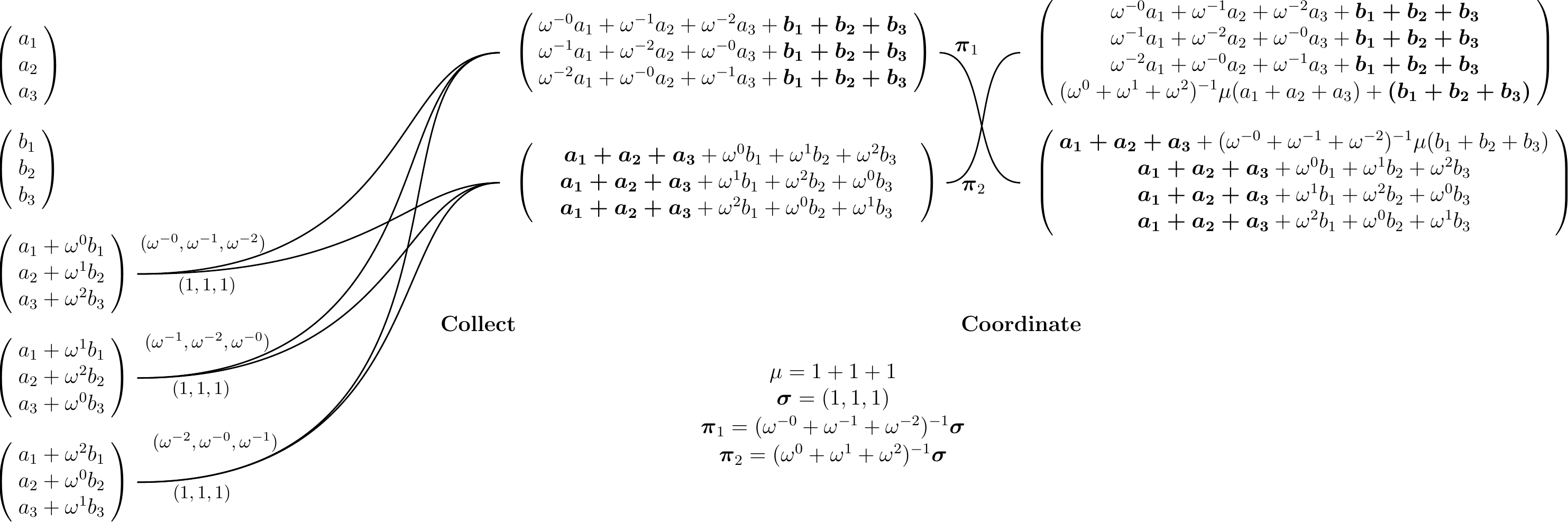}%
	\caption{Exact Repair of the systematic part of an MSCR code ($n=5,k=2,d=3,t=2$). The state of the system after the storing step is not shown but it is clear that the first device can recover $\vec{a}$ and that the second one can recover $\vec{b}$.}
	\label{fig:n5k2t2d3}
\end{figure*}

To this end, it has been proposed to study deterministic schemes, namely exact regenerating codes,  regenerating blocks equal to the lost ones instead of blocks only functionally equivalent. The difference between exact repair and functional repair is shown on Figure~\ref{fig:soa}.  It has been shown that exact repair is strictly harder than functional repair~\cite{Suh2011}, which means that the existence of functional regenerating codes does not imply that exact regenerating codes exist. Hence, an interesting question is whether the previous tradeoffs, which apply to functional repairs, can still be achieved for exact repairs. The problem of repairing exactly a single failure has been well studied~\cite{Dimakis2010b,Wu2009,Shah2010b,Cadambe2010,Rashmi2011,Suh2010a,Suh2011,Papailiopoulos2011}, including intermediary repair schemes such as semi-exact repairs where only a part of the data is regenerated exactly~\cite{Rashmi2009,Wu2011, Cadambe2011,Tamo2011}. However,  the exact repair of multiple failures has been studied mostly according to functional repairs~\cite{NetCod2011,Hu2010,Shum2011} except for the very specific setting $d=k$~\cite{Shum2011,Shum2011b}.  

In this article, we consider only scalar codes where one indivisible sub-block is transmitted between devices during repairs (\emph{i.e.}, $\beta'=1$), thus leading to simpler constructions (Figure~\ref{fig:exact}). When considering the exact repair of single failures, it has been shown that scalar codes are sufficient to construct MBR codes for any value of $n,k,d$ and MSR codes for any values $n,k,d\ge2k-2$. However, scalar codes are not sufficient for building exact scalar MSR codes when $d < 2k-3$~\cite{Shah2010b}. The discussion of vector codes constructions, in which multiple indivisible sub-blocks are transmitted between devices during repairs (\emph{i.e.}, $\beta'>1$) (Figure~\ref{fig:sub}), is deferred to Section~\ref{sec:related} about related work.

In the sequel of the article, we will study the exact repair of regenerating codes when multiple failures occur. We study the non-degenerated case of $d > k$ and use scalar codes ($\beta=1$). We adopt following convention: the data $\vec{v}$ and the codewords $\vec{w}$ are column vectors, the generator matrix $\vec{G}$ is rectangular and the encoding operation $\vec{w}=\vec{G}\vec{v}$ gives a column vector. 

\section{Exact MSCR codes for $k=2$}
\label{sec:code}

In this section, we provide a code construction for scalar MSCR codes supporting exact repairs for $d>k$, $k=2$ and $t=2$.  This code construction also serves as a proof that its possible to repair exactly a $(n,k=2,d=n-t>k,t=2)$ MSCR code.

We consider a system storing a file of size $\MM=k(d-k+t)$ split in $k=2$ blocks $(\vec{a}, \vec{b})$, each of size $\alpha=d-k+t$ sub-blocks. The system consists of $n=d+t$ devices as we assume that all failed devices and all live devices take part to the repair.  In the sequel of the article, we consider a finite field $\mathbb{F}$ having a generator element $\genff$.

The system is compounded of two devices storing the systematic part and $s=n-2$ devices storing the redundancy part.
\begin{itemize}
\item The first systematic device stores $\vec{a}=(a_1,\dots,a_{\alpha})^t$.
\item The second systematic device stores $\vec{b}=(b_1,\dots,b_{\alpha})^t$.
\item The $i$-th redundancy device, $i \in \{0\dots{}\alpha-1\}$ stores \\$\vec{r_i}=(a_1+\genff^{i \mod \alpha}b_1,\dots,a_{\alpha}+\genff^{i + \alpha -1 \mod \alpha}b_\alpha)^{t}$
\end{itemize}
 An example for $k=2$, $d=3$ and $t=2$ is given on Figure~\ref{fig:n5k2t2d3}.

Using the previously defined code, we can state the two following theorems:
\begin{theorem}
It is possible to build minimum storage coordinated regenerating codes that can be repaired exactly when $n=d+t$ (\emph{i.e.}, all devices participate in the repair\footnote{The code we define and the proofs are given for $n=d+t$ for the sake of clarity. However, the method can also be applied to codes where $n>d+t$}), $k=2$ and $t=2$ (\emph{i.e.}, multiple repairs are performed simultaneously).
\label{thm:emscr}
\end{theorem}
\begin{proof}
In the sequel of this section, we review the different properties that are needed for this code to be an MSCR code:
\begin{itemize}
\item It must be an MDS code (\emph{i.e.}, data from any $k=2$ devices must allow recovering the original data).
\item Any two  devices can be repaired exactly.
\end{itemize}
The theorem follows from the code satisfying these properties.
\end{proof}

\begin{theorem}
It is possible to build adaptive regenerating codes that can be repaired exactly when $n=d+t$ (\emph{i.e.}, all devices participate in the repair\footnote{Similarly to~\ref{thm:emscr}, the method can also be applied when $n>d+t$.}), and $k=2$. 
\label{thm:earc}
\end{theorem}
\begin{proof}
In order to show that there exists adaptive regenerating codes~\cite{NetCod2011}, that can be repaired exactly, we need to find a code that has the following properties.
\begin{itemize}
\item It must be an MDS code (\emph{i.e.}, data from any $k=2$ devices must allow recovering the original data).
\item Any two  devices can be repaired exactly.
\item Any single failure can be repaired exactly.
\end{itemize}
The theorem follows from the code satisfying these properties. Note that the two first properties are common with the proof of Theorem~\ref{thm:emscr}.
\end{proof}

\subsection{The MDS property}
This property is trivially satisfied since, when fetching data from any two devices, we get $\alpha$ groups of 2 equations over 2 unknowns, where each group concerns different unknowns. The $i^\mathrm{th}$ group is about $a_i$ and $b_i$ and consists of 2 independent equations. Hence, the unknowns of each group can be recovered and the MDS property is satisfied.

\subsection{Repairing two failures}
The repair consists of the following steps, which map onto the process defined in~\cite{NetCod2011}. In this scheme,  illustrated in Figure~\ref{fig:n5k2t2d3-r12gen}, we do not rely on random linear network coding but give a method for repairing exactly. 

\begin{figure}[htbp]%
\centering
\includegraphics[width=0.9\linewidth]{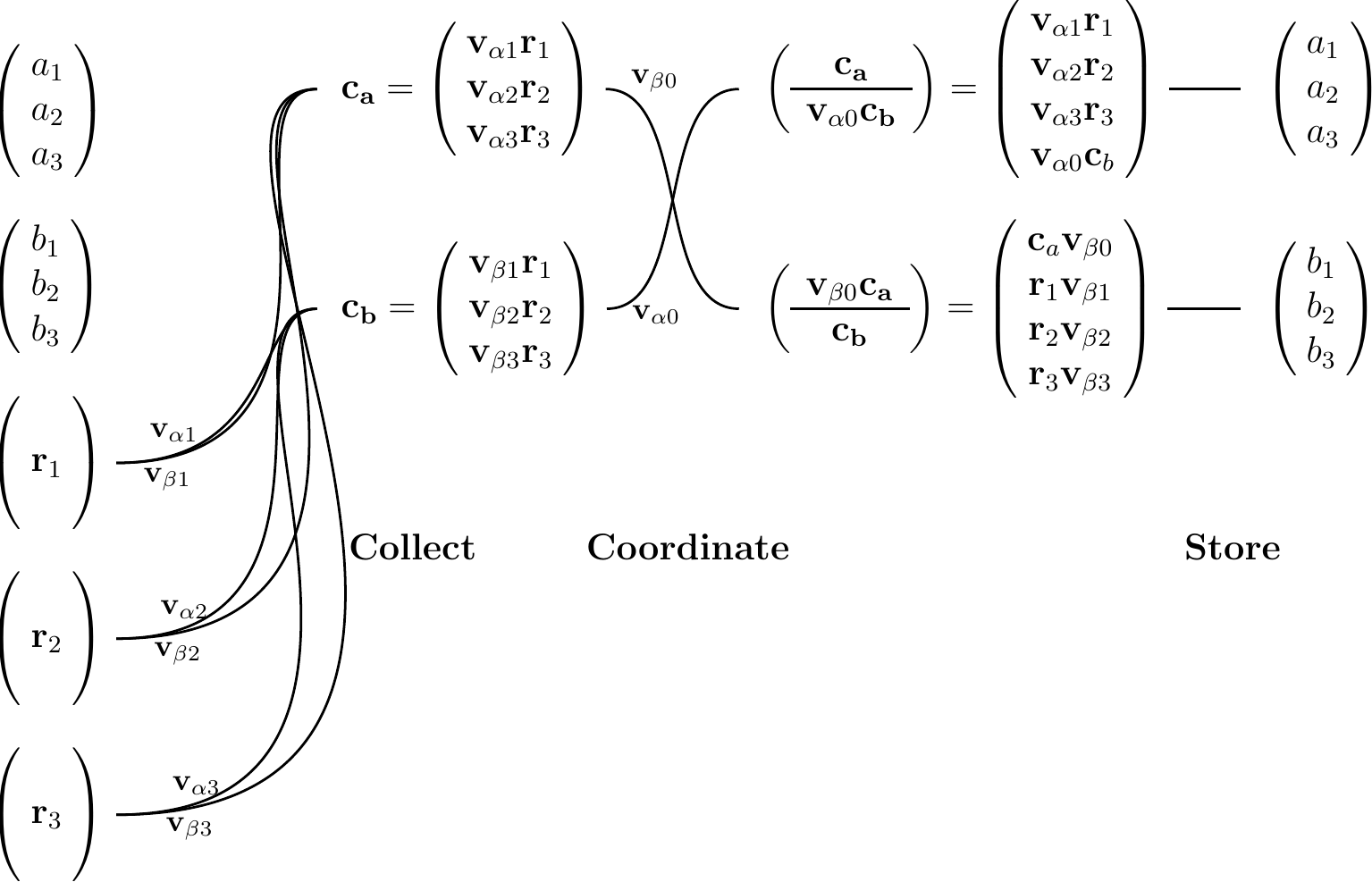}%
\caption{The repair process, with a coordination step. Interfering information transmitted is aligned to allow the recovery of $\vec{a}$ and $\vec{b}$.}
\label{fig:n5k2t2d3-r12gen}
\end{figure}

\begin{IEEEdescription}[\IEEEsetlabelwidth{$\negthickspace\negthickspace\negthickspace\negthickspace$}]%
\item \textbf{1. Identify lost data.} Prior to performing the repair, the system identifies which devices have failed and which blocks have been lost. Given the failure of any two devices (systematic or redundancy), we perform a change of variables to transform the actual code $\mathcal{C}$ into a code $\mathcal{C'}$, in which the failed devices are the systematic ones storing $\vec{a}=(a_1\dots{}a_d)^t$ and $\vec{b}=(b_1\dots{}b_d)^t$. Such a code is guaranteed to exist since the original code is MDS  (same argument as in~\cite{Shah2010b}). Furthermore, the system identifies two spare devices than can host the repaired blocks replacing the lost ones.
\item \textbf{2. Prepare (Collect).} Each live device that participates to the repair computes a sub-block to be sent to the first device and a sub-block to be sent to the second device. All the sub-blocks to be sent to the first device have the common property that the interfering information about $\vec{b}$ is aligned (\emph{i.e.},  the \emph{i}-th live device, storing $\vec{r}_i$,  sends\footnote{In this description, $\vec{v}_{\alpha{}i}\vec{r}_i$, $\vec{w}_{\alpha{}i}\vec{a}$ or $\vec{z}_{\alpha{}i}\vec{b}$ are of scalars (\emph{i.e.}, the resulting matrices are of dimension $1\times{}1$). As a result $\vec{c_a}=(\vec{v}_{\alpha{}1}\vec{r}_1,\dots,\vec{v}_{\alpha{}d}\vec{r}_d)^t$ is a matrix of size $d\times{}1$ and $(\vec{c_a} | c_b\vec{v}_{\alpha{}0})^t$ a matrix of dimension $(d+1) \times 1$.} $\vec{v}_{\alpha{}i}\vec{r}_i=\vec{w}_{\alpha{}i}\vec{a}+\vec{z}_{\alpha}\vec{b}$ so that the spare device receives different information about $\vec{a}$ but the same about $\vec{b}$. To build $\vec{v}_{\alpha{}i}$, given some arbitrary alignment vector  $\vec{z}_\alpha$ and given that $\vec{r}_i=\vec{A}_i\vec{a}+\vec{B}_i\vec{b}$, the repair vector is $\vec{v}_{\alpha{}i}=\vec{z}_\alpha\vec{B}_i^{-1}$. Since the MDS property is satisfied (\emph{i.e.}, we can recover from  $\vec{a}$ and $\vec{r}_i$), $\vec{B}_i$ is invertible, and the repair vector exists. The same applies for $\vec{v}_{\alpha{}0}$ (with $\vec{c_b}=\vec{A}_0\vec{a}+\vec{B}_0\vec{b}$) and $\vec{v}_{\beta{}i}$. The role of $\vec{a}$ and $\vec{b}$ are reversed for sub-blocks to be sent to the second device.
\item \textbf{3. Transfer (Collect).} The sub-blocks prepared are sent and the first (resp. second) spare device stores them temporarily as $\vec{c_a}=(\vec{v}_{\alpha{}1}\vec{r}_1,\dots,\vec{v}_{\alpha{}d}\vec{r}_d)^t$ (resp. $\vec{c_b}$) for further processing during steps 4 and 6.
\item \textbf{4. Prepare (Coordinate).} Using what has been received in step 3, the second spare device prepares a sub-block $\vec{v}_{\alpha{}0}\vec{c_b}=\vec{w}_{\alpha{}0}\vec{a}+\vec{z}_\alpha\vec{b}$ to be send to the first spare device. The interfering information about $\vec{b}$ is aligned as in sub-blocks prepared during step 2. Again, the role $\vec{a}$ and $\vec{b}$ are reversed for the sub-block to be sent from the first to the second spare device.
\item \textbf{5. Transfer (Coordinate).} The sub-blocks prepared are sent and the first (resp. second) spare devices adds them to blocks received in step 3 thus storing $(\vec{c_a} | \vec{v}_{\alpha{}0}\vec{c_b})^t$ (resp. $(\vec{c_b} | \vec{v}_{\beta{}0}\vec{c_a})^t$ ).
\item \textbf{6. Recover and Store.} The $d+1$ sub-blocks $(\vec{c_a} | \vec{v}_{\alpha{}0}\vec{c_b})^t= (\vec{w}_{\alpha_1}\vec{a}+\vec{z}_{\alpha}\vec{b},\dots,\vec{v}_{\alpha{}d}\vec{a}+\vec{z}_{\alpha}\vec{b},\vec{w}_{\alpha{}0}\vec{a}+\vec{z}_{\alpha}\vec{b})^t$ allow recovering both the interfering information received $\vec{w}\vec{b}$ (but not the individual values of $b_i$), and all the desired information $\vec{a} = (a_1\dots{}a_d)^t$ (\emph{i.e.}, the individual values of all sub-block $a_i$) : the received sub-blocks define $d+1$ equations over $d+1$ unknowns $(\vec{z}_{\alpha}\vec{b},a_1,\dots,a_d)$. The lost sub-blocks are thus restored. The second spare device performs a similar processing with the role of $\vec{a}$ and $\vec{b}$ reversed.
\end{IEEEdescription}%

\def\sigmav{\boldsymbol{\sigma}}
\def\piv{\boldsymbol{\pi}}

We now apply this repair method to the code we define, as shown on Figure~\ref{fig:n5k2t2d3}. In order to repair the two systematic devices, during the collecting step, the $i$-th redundancy device sends $(\genff^{-(i \mod \alpha)},\dots,\genff^{-(i+\alpha -1  \mod \alpha)})\vec{r_i}$ to the first device being repaired and $(1\dots1)\vec{r_i}$ to the second device being repaired. The vectors $\vec{v}_{\alpha{}i}$ (resp. $\vec{v}_{\beta{}i}$) are chosen so that $\vec{z}_{\alpha}=\sigmav$ (resp. $\vec{z}_{\beta}=\sigmav$) with $\sigmav=(1\dots1)$. Let us note $\vec{c_a}$ (respectively $\vec{c_b}$) the vector of all $d$ symbols received by the systematic devices repairing $\vec{a}$ (respectively $\vec{b}$).

At the coordination step, the first systematic device sends $(\genff^{-0}+\dots+\genff^{-(\alpha-1)})^{-1}\sigmav\vec{c_a}$ to the second one, while the second one sends $(\genff^{0}+\dots+\genff^{\alpha-1})^{-1}\sigmav\vec{c_b}$ to the first one. 

At the end of these two steps, the first device has received $\alpha+1$ equations. Let us note $\mu=1+\dots+1$ Since all the interfering information about $\vec{b_i}$ is aligned, it can be written as
$$
\left(\begin{array}{c}
\genff^{0}a_1+\dots+\genff^{-(\alpha-1)}a_{\alpha}+\sigmav\vec{b} \\[-1.5eX]
\vdots\\[-0.5eX]
\genff^{-(i \mod \alpha)}a_1+\dots+\genff^{-(i+\alpha -1 \mod \alpha)}a_\alpha+\sigmav\vec{b} \\[-1.5eX]
\vdots\\[-0.5eX]
\genff^{-(\alpha -1 \mod \alpha)}a_1+\dots+\genff^{-(2\alpha-2 \mod \alpha)}a_\alpha+\sigmav\vec{b} \\
(\genff^0+\dots+\genff^{(\alpha-1)})^{-1}\mu(a_1+\dots+a_\alpha)+\sigmav\vec{b}
\end{array}\right)$$

As a consequence, it consists of a system of $\alpha+1$ independent equations and $\alpha+1$ unknowns ($a_i$s and $\sigmav\vec{b}$). As a result, the $\alpha$ unknowns $a_i$ can be recovered. The second device has received something similar with the roles $\vec{a}$ and  $\vec{b}$ exchanged.

This repair process also applies to the repair of redundancy devices. Indeed, during the first step, a change of variables is performed to transform the code $\mathcal{C}$ into a code $\mathcal{C'}$ so that the two redundancy devices (or one redundancy and one systematic device) to be repaired in $\mathcal{C}$ become two systematic devices in $\mathcal{C'}$. Such a code is guaranteed to exist since the original code is MDS~\cite{Shah2010b}. When repairing the 2\nd{} and 3\rd{} devices or the 3\rd{} and 4\th{} devices, the equivalent codes are shown in Figure~\ref{fig:n5k2t2d3-chv}. 

\begin{figure}[htbp]%
\subfloat[devices 2 and 3]{\centering
	\hspace{0.6cm}
	\includegraphics[height=0.7\linewidth]{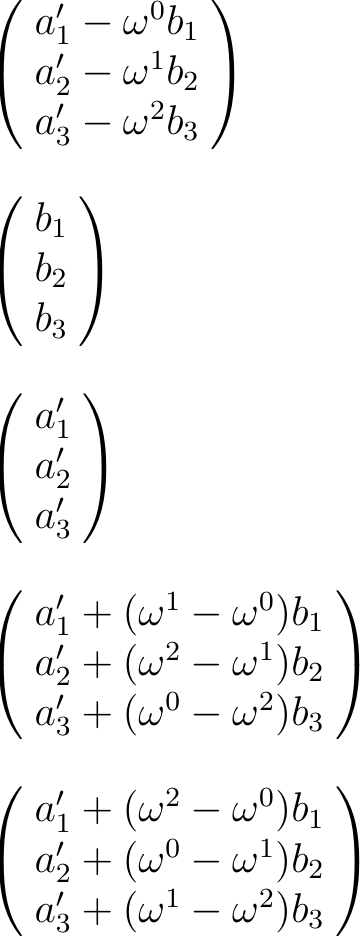}%
	\hspace{0.6cm}
	\label{fig:n5k2t2d3-ns1-chv}%
}
\hfil
\subfloat[devices 3 and 4]{\centering
	\includegraphics[height=0.7\linewidth]{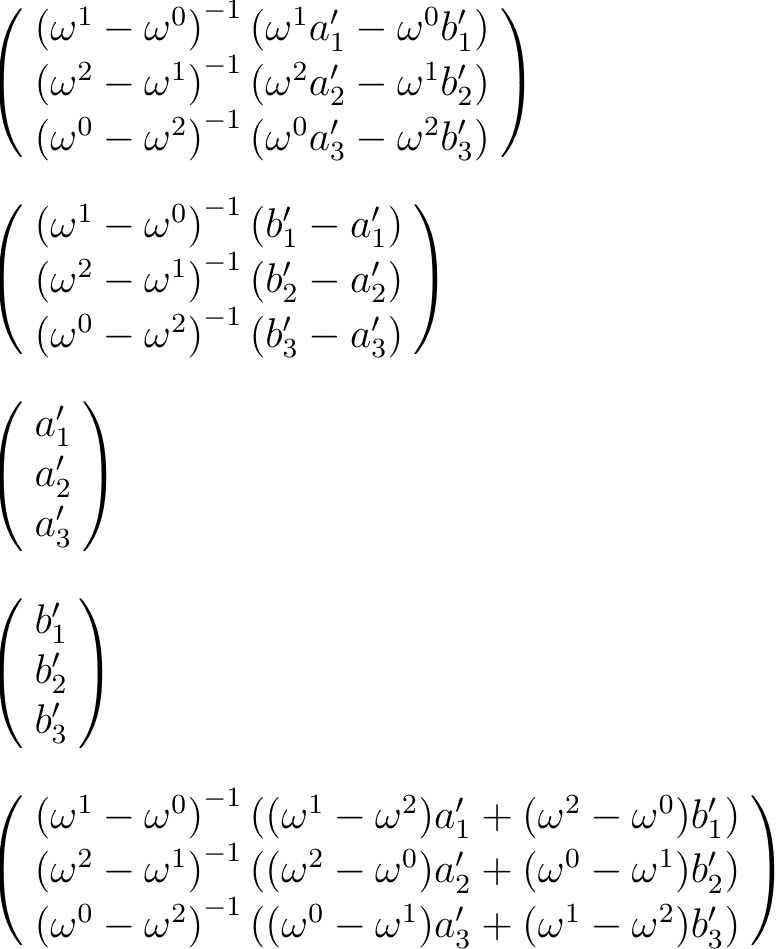}%
	\label{fig:n5k2t2d3-ns2-chv}%
}
\caption{After a change of variable, any two repairs boil down to the repair of two systematics devices. The figure shows the system after a change of variable for the failure of one systematic device and one redundancy device (a) or the failure of two redundancy devices (b).  As a consequence, we can limit our studies to the repair of two systematic devices.}
\label{fig:n5k2t2d3-chv}
\end{figure}

This repair method applied to a code $n=d+t,k=2,d>k,t=2$ ($n=5$ and $d=3$ on Figure~\ref{fig:n5k2t2d3}) naturally extends to other cases such as codes $n>d+t,k=2,d>k,t=2$.

\subsection{Repairing one device}
Finally, repairing one single device is an easier problem, and interference alignment has been used in several codes~\cite{Shah2010b,Suh2011}. However, we need to show that the code construction we present, which support $t=2$, also supports $t=1$ to get exact scalar adaptive regenerating codes. We can apply the same repair method as for repairing two devices except that there is no coordination step and the other systematic device sends directly $\vec{z}_{\alpha}\vec{b} = \sigmav\vec{b}$ during the collecting step. As a result, after the collection step, the failed device has received $\alpha+1$ equations. Since all the interfering information about $\vec{b_i}$ is aligned, it can be written as
$$
\left(\begin{array}{c}
\genff^{0}a_1+\dots+\genff^{-(\alpha-1)}a_{\alpha}+\sigmav\vec{b} \\[-1.5eX]
\vdots\\[-0.5eX]
\genff^{-(i \mod \alpha)}a_1+\dots+\genff^{-(i+\alpha -1 \mod \alpha)}a_\alpha+\sigmav\vec{b} \\[-1.5eX]
\vdots\\[-0.5eX]
\genff^{-(\alpha -1 \mod \alpha)}a_1+\dots+\genff^{-(2\alpha-2 \mod \alpha)}a_\alpha+\sigmav\vec{b} \\
\sigmav\vec{b}
\end{array}\right)$$

As a consequence, it consists of a system of $\alpha+1$ independent equations and $\alpha+1$ unknowns ($a_i$s and $\sigmav\vec{b}$). As a result, the $\alpha$ unknowns $a_i$ can be recovered. 

Since the code we present has the MDS property and supports both repairs of single failures ($t=1$) and repairs of two failures ($t=2$), it implies that it is possible to design exact scalar MSCR codes and exact scalar adaptive regenerating codes, thus leading to Theorems~\ref{thm:emscr} and~\ref{thm:earc}.

\section{Impossibility of Independent Interference Alignment for Exact MSCR when $k\ge 3$}
\label{sec:nonachiev}
In this section, we examine whether the previous scheme, inspired by the repair of single failures~\cite{Shah2010b,Suh2011}, can be applied to multiple failures when $k \ge 3$.

When repairing a single failed systematic\footnote{Again, the repair of a redundancy block in a code $\mathcal{C}$ is equivalent to the repair of systematic block in a code $\mathcal{C'}$.} block $\vec{a}$, the information about the $k-1$ other systematic blocks must be aligned, as shown in~\cite{Shah2010b}. In particular, it is required that blocks are aligned \emph{independently}. Indeed, if we consider that the systematic devices send vectors $\vec{v}_{\beta}\vec{b}$, $\vec{v}_{\gamma}\vec{c}$\dots, and that the $i$-th redundancy device sends $\vec{v}_{\alpha{}i}\vec{a}+\vec{v}_{\beta{}i}\vec{b}+\vec{v}_{\gamma{}i}\vec{c}\dots$, to the device repairing $\vec{a}$, then it must be that, for all $i$, $\colspan{(\vec{v}_{\beta{}i})} = \colspan{(\vec{v}_{\beta})}$, $\colspan{(\vec{v}_{\gamma{}i})} = \colspan{(\vec{v}_{\gamma})\dots}$ (\emph{i.e.}, systematic blocks are considered independently and all the information about each interfering block received at the device performing the repair span only one dimension).

We show that under this requirement, exact repair is not possible if $k \ge 3$. We give a first proof, and explain the meaning of this impossibility on the information flow graph~\cite{Dimakis2010,NetCod2011}.

\begin{figure*}[t]%
\includegraphics[width=\linewidth]{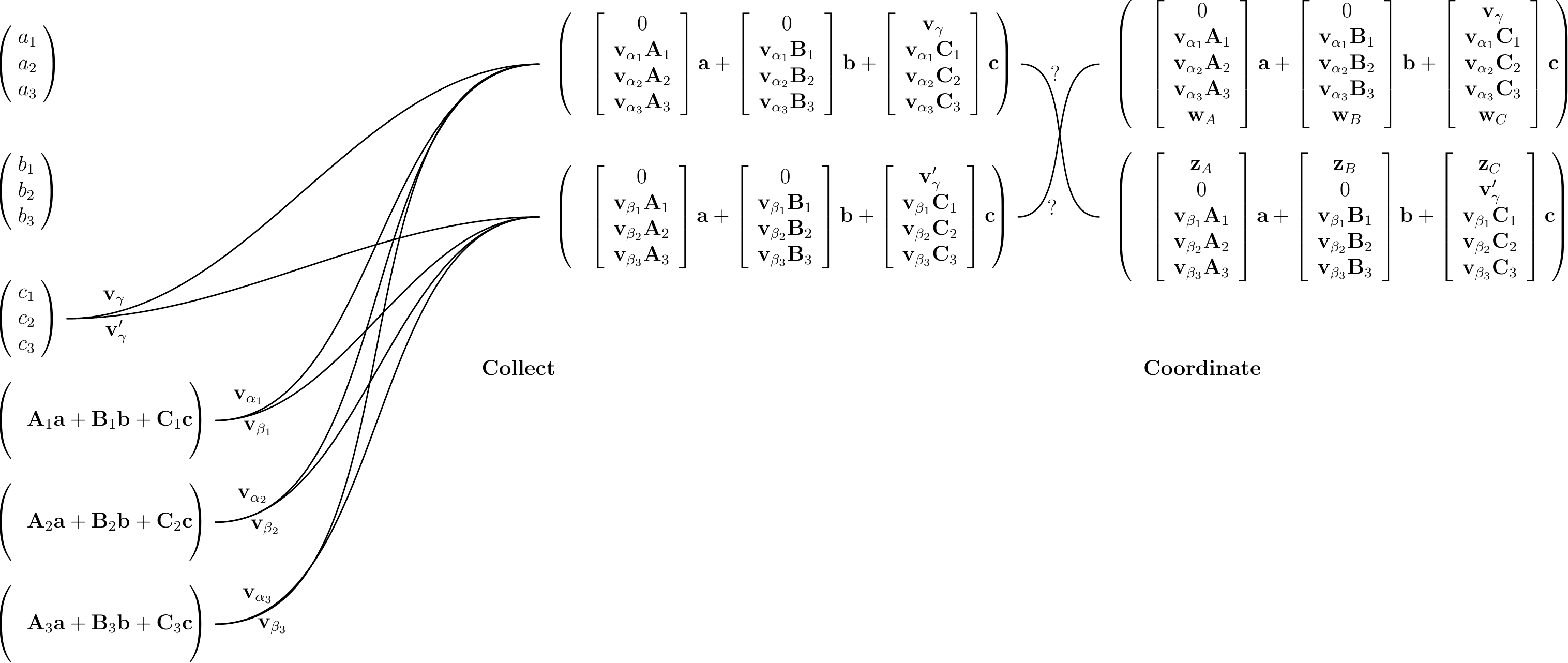}%
\caption{Impossibility of achieving exact repair of the systematic part of an MSCR code ($k \ge 3$, $d > k$ and $t \ge 2$)}%
\label{fig:n6k3t2d4}%
\end{figure*}%

\begin{theorem}
When requiring interference alignment to be applied independently on all devices, it is not possible to repair exactly MSCR codes with $k \ge 3$ and $t \ge 2$ in the scalar case (\emph{i.e.}, $\MM = k(d-k+t)$ such that each device stores only $d-k+t$ sub-blocks of size $\beta=1$). 
\end{theorem}

\begin{proof}
Since any MDS code $\mathcal{C}$ can be turned into a equivalent systematic code $\mathcal{C'}$ (as explained in \cite{Shah2010b}), we base our proof on Lemma~\ref{thm:lemsys}. Indeed, if it was possible to repair exactly MSCR codes  with $k \ge 3$ and $t \ge 2$, it would be possible to build systematic MSCR codes that can be repaired exactly.
\end{proof}

\begin{corollary}
When requiring interference alignment to be applied independently on all devices, it is not possible to repair exactly adaptive regenerating codes with $k \ge 3$  in the scalar case (\emph{i.e.}, $\MM = k(d-k+t)$ such that each device stores only $d-k+t$ sub-blocks of size $\beta=1$). 
\end{corollary}

\begin{proof}
Since the repair of  adaptive regenerating codes with $k \ge 3$ and $t \ge 2$ is very similar to the repair of of MSCR codes, the impossibility result also applies to adaptive regenerating codes. In particular, exact MSCR codes could be derived from exact adaptive regenerating codes by fixing values of $d$ and $t$ if such adaptive regenerating codes existed.
\end{proof}

\begin{lemma}
When requiring interference alignment to be applied independently on all devices, it is not possible to repair exactly systematic MSCR codes with $k \ge 3$ and $t \ge 2$ in the scalar case (\emph{i.e.}, $\MM = k(d-k+t)$ such that each device stores only $d-k+t$ sub-blocks of size $\beta=1$). 
\label{thm:lemsys}
\end{lemma}

\begin{IEEEproof}
Let us consider a code with $k \ge 3$, $t \ge 2$, $d > k$ , $n \ge d+t$  and $\alpha=d-k+t$. Let us assume that we want independent interference alignment (\emph{i.e.}, each interfering block spans only a sub-space of dimension 1).

The $k$ first devices store systematic blocks as vectors $\vec{a}=(a_i)_{1\le{}i\le{}\alpha}$, $\vec{b}=(b_i)_{1\le{}i\le{}\alpha}$, $\vec{c}=(c_i)_{1\le{}i\le{}\alpha}$\dots The $n-k$ remaining devices store redundancy blocks as $\vec{r}_j=\vec{A}_1\vec{a}+\vec{B}_1\vec{b}+\vec{C}_1\vec{c}+\dots$. Thus leading to a set-up similar to the one depicted on Figure~\ref{fig:n6k3t2d4}.

We are going to proof, by contradiction, that exact repairs of systematic codes in the scalar case (\emph{i.e.}, $\beta=1$) are not achievable when  $k \ge 3$ and $t \ge 2$.  For the sake of clarity, our proof will describe the case of $t=2$, $k=3$ and $d=4$ but it naturally extends to any larger values. 

Assume that it is possible to repair exactly. Hence, it is possible to repair the simultaneous failure of devices storing $\vec{a}$ and $\vec{b}$. We consider this case  and  examine how exact repairs constraint the system.

For each device being repaired, all live devices project what they store onto a single vector and send this vector to the said device being repaired. Then, the devices being repaired coordinate by exchanging a single vector (a projection of what they have received so far).
Hence, the device repairing $\vec{a}$ receives, at the end of both the collecting step and the coordination step:
\begin{equation}
\matdesc{ 0 \\ \vec{v}_{\alpha_1}\vec{A}_1 \\ \vec{v}_{\alpha_2}\vec{A}_2 \\ \vec{v}_{\alpha_3}\vec{A}_3 \\ \vec{w}_A }(\vec{a})+
\matdesc{ 0 \\ \vec{v}_{\alpha_1}\vec{B}_1 \\ \vec{v}_{\alpha_2}\vec{B}_2\\ \vec{v}_{\alpha_3}\vec{B}_3\\ \vec{w}_B}(\vec{b})+
\matdesc{ \vec{v}_\gamma \\ \vec{v}_{\alpha_1}\vec{C}_1 \\ \vec{v}_{\alpha_2}\vec{C}_2 \\ \vec{v}_{\alpha_3}\vec{C}_3 \\ \vec{w}_C}(\vec{c})
\end{equation}

To be able to recover $\vec{a}$, we must be able to decode the $d-k+t=3$ desired unknows of $\vec{a}$ out of the $d+t-1=5$ equations containing a total of $k(d-k+t)=9$ unknowns. Hence, when aligning independently we must have, 
\begin{align}
\rank\left({\matdesc{ \vec{v}_\gamma \\ \vec{v}_{\alpha_1}\vec{C}_1\\ \vec{v}_{\alpha_2}\vec{C}_2\\ \vec{v}_{\alpha_3}\vec{C}_3 \\ \vec{w}_C}()}\right) = 1\textrm{,}
&&
\rank\left({\matdesc{ 0 \\ \vec{v}_{\alpha_1}\vec{B}_1 \\ \vec{v}_{\alpha_2}\vec{B}_2 \\ \vec{v}_{\alpha_3}\vec{B}_3\\\ \vec{w}_B}()}\right) = 1
\label{eq:aligna}
\end{align}
and, 
\[
\rank\left({\matdesc{  0 \\ \vec{v}_{\alpha_1}\vec{A}_1 \\ \vec{v}_{\alpha_2}\vec{A}_2 \\ \vec{v}_{\alpha_3}\vec{A}_3 \\ \vec{w}_A}()}\right) = 3
\]

  Similarly, to be able to recover $\vec{b}$, we must have,  
\begin{align}
\rank\left({\matdesc{ \vec{z}_{C} \\ \vec{v}'_\gamma \\ \vec{v}_{\beta_1}\vec{C}_1\\ \vec{v}_{\beta_2}\vec{C}_2\\ \vec{v}_{\beta_3}\vec{C}_3}()}\right) = 1\textrm{,}
&&
\rank\left({\matdesc{\vec{z}_{A} \\ 0 \\ \vec{v}_{\beta_1}\vec{A}_1 \\ \vec{v}_{\beta_2}\vec{A}_2 \\ \vec{v}_{\beta_3}\vec{A}_3}()}\right) = 1
\label{eq:alignb}
\end{align}
and, 
\[
\rank\left({\matdesc{ \vec{z}_{B}\\  0 \\ \vec{v}_{\beta_1}\vec{B}_1 \\ \vec{v}_{\beta_2}\vec{B}_2 \\ \vec{v}_{\beta_3}\vec{B}_3 }()}\right) = 3
\]

Let us consider the choice of vectors $\vec{v}_\gamma$, $\vec{v}_{\alpha_i}$, $\vec{v}_{\beta_i}$   and of matrices $\vec{C}_i$ that allows exact repairs (\emph{i.e.}, such that constraints on ranks are satisfied) with coordination (\emph{i.e.}, $k \ge 3$ and $t \ge 2$):
\begin{itemize}
\item All $\vec{v}_{\alpha_i}\vec{C}_i$ must be collinear according to \eqref{eq:aligna}.
\item All  $\vec{v}_{\beta_i}\vec{C}_i$ must be collinear too according to \eqref{eq:alignb}.
\item During the coordination step, what is sent by the device repairing $\vec{a}$ will necessarily be collinear to $\vec{v}_{\alpha_i}\vec{C}_i$  (\emph{i.e.}, what is stored) and  to vector $\vec{v}_{\gamma}$. Let us name this vector, which is colinear to $\vec{v}_{\gamma}$, $\vec{z}_C$. According to \eqref{eq:alignb}, $\vec{z}_{C}$, and hence $\vec{v}_{\gamma}$  must be colinear to all  $\vec{v}_{\beta_i}\vec{C}_i$. Hence, we have: 
$\forall{} i$,  $\vec{v}_{\alpha_i}=\nu_i\vec{v}_{\gamma}\vec{C}^{-1}_i$ and $\vec{v}_{\beta_i}=\mu_i\vec{v}_{\gamma}\vec{C}^{-1}_i$. Note that the matrix $\vec{C}_i$ is invertible to guarantee the MDS property.
\end{itemize}

As a result, for all $i \in \{1\dots{}d\}$, vectors $\vec{v}_{\alpha_i}$ and $\vec{v}_{\beta_i}$ are collinear since 
\begin{equation}
\vec{v}_{\alpha_i} = \frac{\nu_i}{\mu_i}\vec{v}_{\beta_i}
\label{eq:colin}
\end{equation}

Let us consider the choice of matrices for $\vec{B_i}$ that allows exact repairs on the device repairing $\vec{a}$. According to~\eqref{eq:aligna}, we must have  $\rank{ (\vec{B}_1\vec{v}_{\alpha_1},\dots,\vec{B}_d\vec{v}_{\alpha_d})^t } = 1$,  which is equivalent to:
\begin{equation}
\rho_1\vec{v}_{\alpha_1}\vec{B}_1 = 
\rho_2\vec{v}_{\alpha_2}\vec{B}_2 = 
\dots = 
\rho_d\vec{v}_{\alpha_d}\vec{B}_d
\label{eq:aligned}
\end{equation}

Combining \eqref{eq:colin} and \eqref{eq:aligned}, we can deduce that 
\begin{equation}
\rho_1\frac{\nu_1}{\mu_1}\vec{v}_{\beta_1} \vec{B}_1= 
\rho_2\frac{\nu_2}{\mu_2}\vec{v}_{\beta_2}\vec{B}_2 = 
\dots = 
\rho_d\frac{\nu_d}{\mu_d}\vec{v}_{\beta_d}\vec{B}_d
\end{equation}

As a result, $\rank{ (\vec{B}_1\vec{v}_{\beta_1},\dots,\vec{B}_d\vec{v}_{\beta_d})^t } = 1$ which is in contradiction with the hypothesis~\eqref{eq:alignb}, that $\vec{b}$ can be repaired too (\emph{i.e.}, $\rank{ (\vec{B}_1\vec{v}_{\beta_1},\dots,\vec{B}_d\vec{v}_{\beta_d})^t } \ge d-1$)
Hence, the exact repair of two failed devices when $k>3$ is impossible.

A rather similar proof can be performed assuming that $\vec{c}$ is being recovered too ($t=k$). In this case, the vector about $\vec{c}$ being sent to the devices repairing $\vec{a}$ and $\vec{b}$ during the coordination step needs to be colinear too. Thus leading to the same conclusion that the system is over constrained.

The proof naturally extends to any higher value of $k$ and $t$. Hence, repairing exactly with $d>k$ and $t > 2$ is impossible in the case of scalar codes (\emph{i.e.}, $\beta=1$) based on independent interference alignment.

\end{IEEEproof}

This impossibility means that at some point, the amounts of information that goes through the information flow graph~\cite{Dimakis2010,NetCod2011} is too low.  Indeed, to ensure that the file is kept over time, all cuts between the source $S$ and any data collector $DC$ in a graph representing the transfer of data between devices during repairs must be greater than or equal to $\mathcal{M}$~\cite{NetCod2011}. However, if we consider the graph of Figure~\ref{fig:impgraph} and force the device storing $c$ to send the same $\beta$ bits of information (by requiring alignment of the information) to both the device storing $\vec{a}$ and the device storing $\vec{b}$, then the cut shown on the graph of Figure~\ref{fig:impgraph} has an insufficient capacity of $8\beta<\mathcal{M}$.

\begin{figure}[h]
	\centering%
	\includegraphics[width=0.85\linewidth]{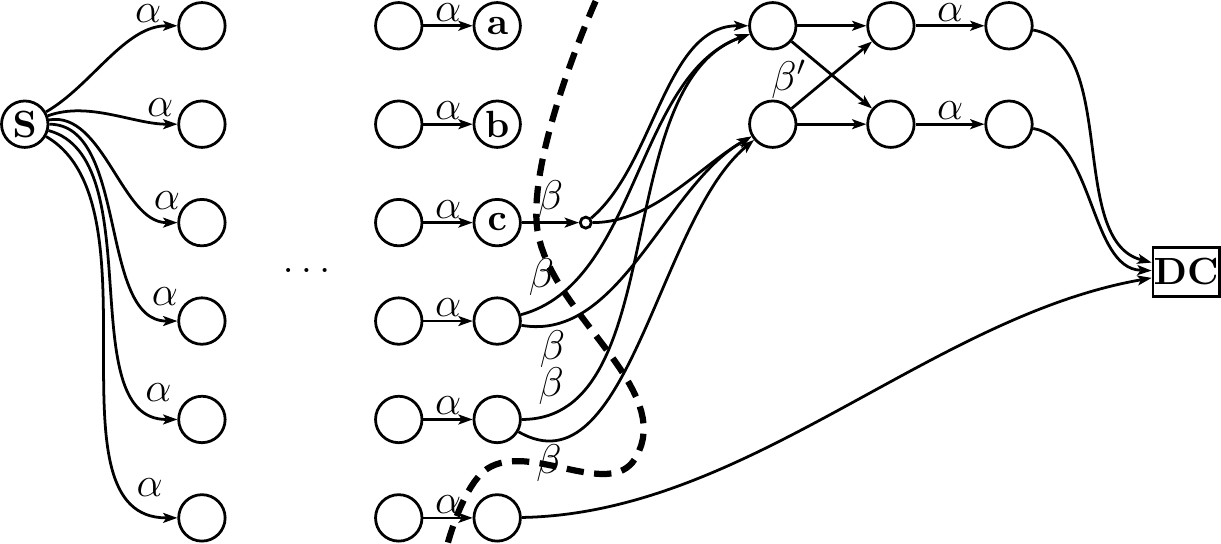}%
	\caption{If the system is constrained so that the third (or any additional) interfering device sends the same information to all devices because of alignment constraints, the flow that can go through the network is no longer equal to the file size $\mathcal{M}$. In the example of the Figure, where $\mathcal{M}=9\beta$ and $\alpha=3\beta$, the capacity of the cut shown is only $\alpha + 5\beta = 8 \beta < \mathcal{M}$. As a result, the amounts of information that go through the network are not sufficient.}%
	\label{fig:impgraph}%
\end{figure}

Interference alignment aims at encoding transmitted data such that all interferences at the receiver (\emph{i.e.}, undesired signals) are perfectly aligned and do not inhibit the reception of the desired signal. In the context of wireless, the channel matrices defining the transmission are imposed by nature and encoding matrices are carefully chosen to achieve interference alignment. When considering regenerating codes for single failures, both the channel matrices and the encoding matrix can be chosen, but it is required that one single encoding matrix allows for interference alignment at multiple receivers, each receiver acquiring a different signal. Yet, this allows a wide set of parameters to be considered. However, with coordinated regenerating codes relying on independent interference alignment, where at least two devices need to coordinated, any undesired signal from any third device must be aligned in the same way on the first and second device that coordinate. As we have just shown, this over-constrains the system thus exact repairs using independent interference alignment are not possible.

\begin{figure*}[t]%
\includegraphics[width=\linewidth]{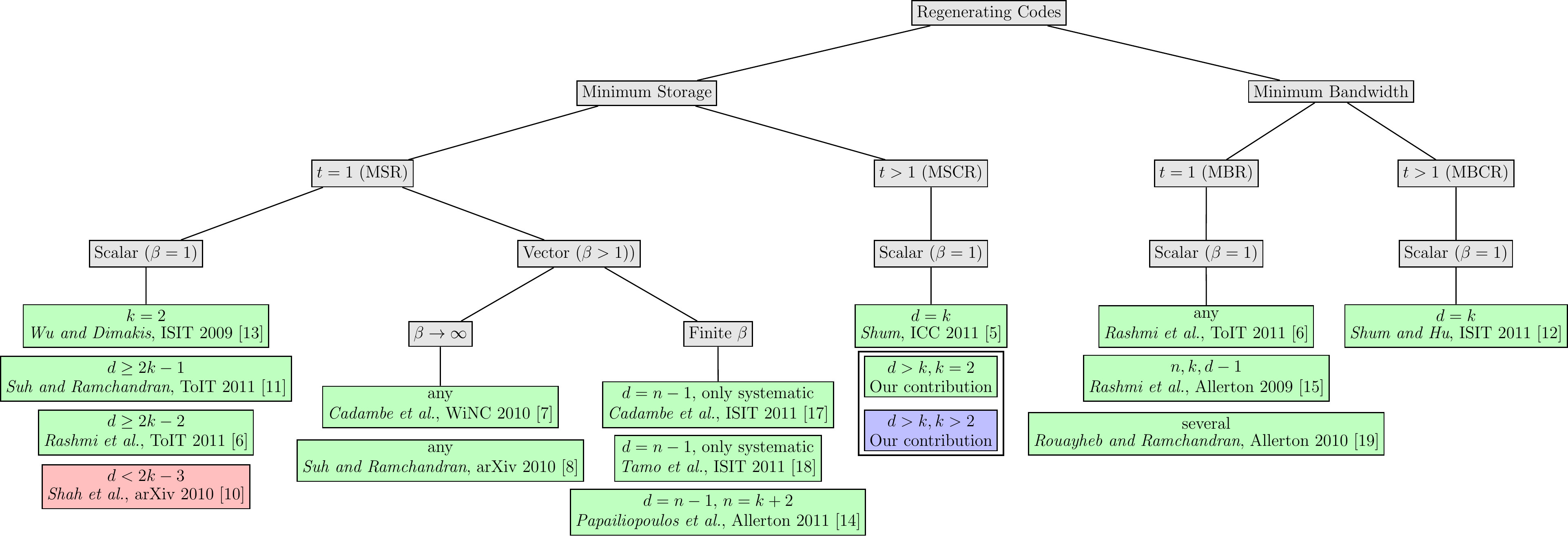}%
\caption{Current state of art results for exact repair codes. It compiles the most recent code constructions and impossibility results. Most codes presented here achieve exact repair for both the systematic and the redundancy part. The contributions of this paper concern scalar MSCR codes for multiple repairs (\emph{i.e.}, $t>1$, $\beta=1$ and $d>k$) and are surrounded by the black box.}%
\label{fig:soatree}%
\end{figure*}%
\section{Related Work}
\label{sec:related}
Figure~\ref{fig:soatree} gives an overview of results related to the construction of exact regenerating codes. Green nodes in the tree corresponds to achievability results while red nodes indicate that it has been shown that it is not possible to build codes for the specified parameters. The blue node correspond to an impossibility in some cases. Two main classes of codes exist, namely scalar and vector codes.  Scalar codes rely on indivisible sub-blocks of size $\beta=1$ as shown on Figure~\ref{fig:exact}. Yet, scalar codes are not always sufficient as explained hereafter. Hence vector codes, relying on sub-packetization, have been defined. In these codes, manipulated sub-blocks are smaller than the smallest amount of information to be transmitted (\emph{i.e.}, sub-blocks are of size $\frac{\beta}{r}$ such that to $r$ indivisible sub-blocks are transmitted when sending $\beta=r$) as shown on Figure~\ref{fig:sub} where $\beta=2$.

Among all possible regenerating codes, most of the studies have focused on the minimum storage point. For MSR codes that are able to repair single failures ($t=1$), studies have heavily relied on interference alignment, first applied to $k=2$ in \cite{Wu2009}. The best known scalar codes either use interference alignment~\cite{Suh2011} to allow  $d \ge 2k-1$, or use the product matrix framework~\cite{Rashmi2011} to allow $d \ge 2k-2$.  However, scalar codes cannot be used to achieve $d <  2k-3$ as shown in~\cite{Shah2010b}. 

To circumvent this impossibility of constructing scalar MSR codes when $d<2k-3$, it has been proposed to rely on vector codes (\emph{i.e.}, $\beta>1$). Vector codes supporting exact repair can be built for any values $n,k,d$ when $\beta  \rightarrow \infty$~\cite{Cadambe2010,Suh2010a}. However, these constructions require infinite sub-packetization and, hence, are not practical.  Recent works~\cite{Cadambe2011, Tamo2011} have shown that finite sub-packetization $\beta=(n-k)^k$  is sufficient to perform exact repair of the systematic devices leading to practical codes. The repair of all devices is possible when $d=n-1, n=k+2$ as shown in~\cite{Papailiopoulos2011}. As a result, the exact repair of all devices with vector MSR codes is not fully solved.

For the case of multiple failures $t>1$, only scalar MSCR codes ($\beta=1$) have been considered. Previous work~\cite{Shum2011} only considered the degenerated case of $d=k$ where the costs of coordinated/cooperative regenerating codes is equivalent to the costs of erasure correcting codes with lazy repairs. In this work, where $\alpha=t$, the repair boils down to repairing  in parallel $t$  independent erasure correcting codes (\emph{i.e.}, no network coding is needed). The work we present in this paper is the first to consider a non-degenerated case $d>k$ and to apply interference alignment when multiple failures are repaired simultaneously leading to the codes we define in Section~\ref{sec:code}, which are restricted to $k=2$. Furthermore, in Section~\ref{sec:nonachiev}, we show that independant interference alignment with scalar codes is not sufficient for building exact MSCR codes when $k \ge 3$.
 
With respect to the MBR point, the best known construction~\cite{Rashmi2011} are scalar codes based on the product matrix framework and allow the repair for any value of $n,k,d$. Some interesting alternative codes~\cite{Rashmi2009,Rouayheb2010} allow repair by transfer (\emph{i.e.}, without performing any linear operation) and rely on fractional repetition codes. 

When multiple failures are repaired simultaneously, the only MBCR codes again consider the case of $d=k$ and map to repairing $t$ independant erasure correcting codes~\cite{Shum2011b}. The existence of MBCR codes when $d>k$ remains an open question.

Finally, regenerating codes~\cite{Dimakis2007,Dimakis2010} can be extended into adaptive codes~\cite{NetCod2011,Wang2010} that support dynamic systems. The first supports repairing multiple failures optimally and has a constant $\beta$ as long as $n=d+t$ (\emph{i.e.}, as long as the total system size $n$ including both live devices and failed devices being repaired remains constant) that makes practical implementation easier~\cite{Kermarrec2010INRIA}. These codes are highly related to minimum storage codes. In particular, the existence (resp. non-existence) of exact adaptive regenerating codes is strongly tied to the existence (resp. non-existence) of exact MSCR codes. In particular, our exact MSCR codes of Section~\ref{sec:code} are also adaptive regenerating codes, and the impossibility shown in Section~\ref{sec:nonachiev} also applies to exact adaptive regenerating codes.

\section{Conclusion}
In this paper, we applied independent interference alignment to minimum storage coordinated regenerating codes (MSCR) and show that this technique allows exact repair if and only if $k=2$. Our results also apply to adaptive regenerating codes thus providing an interesting solution for the implementation of practical systems when $k=2$.

To overcome the impossibility shown in this paper, several tracks can be considered: \emph{(i)} considering a technique that does not align the interferences independently, \emph{(ii)} building vector codes (\emph{i.e.}, relying on sub-packetization with $\beta > 1$ by opposition to scalar codes $\beta=1$ considered in this paper  as done in \cite{Cadambe2010,Suh2010a,Cadambe2011,Tamo2011}), or \emph{(iii)} building minimum bandwidth coordinated regenerating codes (MBCR) (for single failure, codes exist for all parameters~\cite{Rashmi2011}).
 Finally, the related question of achievable limits for high rate exact MSCR when relying on scalar codes remains open.

\bibliographystyle{IEEEtran}
\bibliography{regenerating}

\end{document}